\documentclass[12pt]{article}
\usepackage{amssymb}
\usepackage{amsmath}
\usepackage{mathtools}
\usepackage{geometry}
\usepackage{setspace}
\usepackage{natbib}
\usepackage{indentfirst}
\usepackage{array,epsfig,fancyheadings,rotating}
\usepackage{verbatim}
\usepackage{sectsty, secdot}
\usepackage[hidelinks]{hyperref}  
\usepackage{caption}
\usepackage{enumitem}
\usepackage{url}
\usepackage{bm}
\usepackage{color}
\newtheorem{theorem}{Theorem}[section]

\newtheorem{lemma}{Lemma}[section]

\newenvironment{proof}[1][Proof]{\noindent\textbf{#1.} }{\ \rule{0.5em}{0.5em}}
\sectionfont{\fontsize{12}{14pt plus.8pt minus .6pt}\selectfont}

\subsectionfont{\fontsize{12}{14pt plus.8pt minus .6pt}\selectfont}
\newcommand*{\SuperScriptSameStyle}[1]{%
  \ensuremath{%
    \mathchoice
      {{}^{\displaystyle #1}}%
      {{}^{\textstyle #1}}%
      {{}^{\scriptstyle #1}}%
      {{}^{\scriptscriptstyle #1}}%
  }%
}

\newcommand*{\oneS}{\SuperScriptSameStyle{*}}
\newcommand*{\twoS}{\SuperScriptSameStyle{**}}
\newcommand*{\threeS}{\SuperScriptSameStyle{*{*}*}}

\usepackage{soul}
\def\zhao{\textcolor{blue}}

\begin{document}

\title{A Sequential Learning Procedure with Applications to Online Sales Examination} \author{ Jun Hu\footnote{Jun Hu is an Assistant Professor in the Department of Mathematics and Statistics, Oakland University, 146 Library Drive, Rochester, MI 48309, USA. Tel.:~248-370-3434. Email address: junhu@oakland.edu.}, Yan Zhuang\footnote{Yan Zhuang is an Assistant Professor in the Mathematics and Statistics Department, Connecticut College.}, Shunan Zhao\footnote{Shunan Zhao is an Associate Professor in the Department of Economics, Oakland University.}  
}

\date{}

\maketitle

\bigskip

\begin{abstract}

In this paper, we consider the problem of estimating parameters in a linear regression model. We propose a sequential learning procedure to determine the sample size for achieving a given small estimation risk, under the widely used Gauss-Markov setup with independent normal errors. The procedure is proven to enjoy the second-order efficiency and risk-efficiency properties, which are validated through Monte Carlo simulation studies. Using e-commerce data, we implement the procedure to examine the influential factors of online sales.

\bigskip


\bigskip

\end{abstract}

\bigskip


\setcounter{section}{0}
\setcounter{equation}{0}
\section{Introduction}\label{Sect. 1}

Regression modeling is a statistical technique for investigating the relationship between variables. There is a vast application of regression analysis in almost every field such as biology, chemistry, psychology, physics, and economics. 

In order to estimate the linear model parameters, the least-squared estimation method is widely used. The estimation is commonly conducted using existing available archival data. However, in studies involving data retrieval, it is often the case that one may need to decide, at least roughly, the size of the sample, in order to have an idea of how long the data collection will take and how much money or how many resources it will need. Then, a natural question to ask is how many observations we would need in order to have an accurate or reliable estimation of the linear model parameters. Facing such a question is becoming more common nowadays as the concept of ``big data" is starting to gain traction among researchers in all fields. While textual, audio, and visual data collection through web scrapping is entering the toolbox of social science scholars, data retrieval can be an essential part of research.

To answer this question, sequential estimation procedures emerge. The idea of sequential estimation was first brought up by \citeauthor{Stein (1945)} (\citeyear{Stein (1945)}, \citeyear{Stein (1949)}) as a solution to figure out the sample size for estimating a normal mean with a fixed margin of error when the variance remains unknown. This is known as Stein's two-stage sampling procedure which requires two samples: a pilot sample and a second-stage sample whose collection is determined by the pilot sample. Afterwards, many sequential procedures were developed. Purely sequential sampling methodologies were proposed for the problem of constructing fixed-width confidence intervals in \cite{Anscombe (1953)} and \cite{Chow and Robbins (1965)}. The idea was to sample one item at a time until it achieved a certain accuracy requirement. In a slightly different direction, a minimum risk point estimation (MRPE) problem was formulated by \cite{Robbins (1959)}, in which a purely sequential sampling methodology was put forward to estimate an unknown normal mean when the variance was assumed unknown, under the absolute error loss plus a linear cost of sampling.

Modified and innovative sequential methodologies have been continuously developed after the aforementioned studies. \cite{Hall (1981)} developed a three-stage methodology which was named \textit{triple sampling} for mean estimation. He then built ideas on an accelerated sequential rule and outlined some basic properties in his later work \citep{Hall (1983)}. \cite{Mukhopadhyay and Solanky (1991)} further investigated the accelerated sequential idea and developed some unified rules. They claimed that accelerated sequential methodologies could save a considerable amount of sampling operations and often enjoy similar properties as those of purely sequential procedures. Most recently, \cite{Mukhopadhyay and Wang (2020a)} brought up a new sequential sampling scheme that considered recording $k$ observations at a time. \cite{Hu (2020)} developed a double-sequential sampling scheme that respectively samples $k$ observations in the first step and then one observation at a time. \cite{Hu and Zhuang (2022)} proposed a general sequential sampling scheme, which incorporates four different types of sampling procedures: the purely sequential sampling procedure; the ordinary accelerated sequential sampling procedure; the k-at-a-time purely sequential sampling procedure; and the k-at-a-time accelerated sequential sampling procedure. 

In this paper, we will focus on developing a sequential learning procedure for estimating linear regression parameters with a bounded estimation risk. We will provide a general sequential sampling scheme so that one can customize it to a certain sampling scenario. For example, if the sample items come in batches, one may take a bunch of items at a time; if one wants to save time for sampling logistics, one may collect a relatively small number of items one by one to get the first part of the sample, and then make good the shortfall of the projected sample size by collecting the second part of the sample all in one batch. We show the finite-sample performance of the proposed learning procedure using a set of simulations. In a real-data application, we showcase the learning procedure by applying it to study factors that affect online sales of electronic products on Tmall.com, one of the largest online commerce platforms in China. Given that all sales data need to be extracted through Tmall's official website using a program on a daily basis, the application provides a good illustration of the usefulness of our learning method to determine the required sample size.

The rest of the paper is organized as follows. Section \ref{Sect. 2} brings up the idea of point estimation of linear model coefficients using sequential learning procedures. Section \ref{Sect. 3} summarizes our simulation studies. A real data illustration is discussed in Section \ref{Sect. 4}. Section \ref{Sect. 5} lays down essential proofs of the theoretical results. Section \ref{Sect. 6} concludes the paper with some final thoughts.



\setcounter{equation}{0}
\section{Sequential Point Estimation in a Linear Model}\label{Sect. 2}

Suppose that we have $p(\ge1)$ explanatory variables, denoted by $\bm{X}=(X_1,...,X_p)$, which leads to the response $Y_i$ when
$$X_1 = x_{i1},...,X_p=x_{ip}, \ i=1,2,...$$
Having recored $n$ independent observations
$$\left( Y_i,x_{i1},...,x_{ip}\right), \ i=1,2,...,n,$$
let us consider the following standard linear model
\begin{equation}\label{lm}
{\bm{Y}}_n = {\mathbb{X}}_{n}\bm{\beta} + {\bm{\varepsilon}}_{n},
\end{equation}
where ${\bm{Y}}_n = (Y_1,Y_2,...,Y_i,...,Y_n)^{\prime} \in \mathbb{R}^n$ is the response vector,
$${\mathbb{X}}_{n} = 
\begin{pmatrix}
x_{11} & \cdots & x_{1p}\\
x_{21} & \cdots & x_{2p}\\
\vdots &        & \vdots\\
x_{i1} & \cdots & x_{ip}\\
\vdots &        & \vdots\\
x_{n1} & \cdots & x_{np}\\
\end{pmatrix}
_{n\times p} \in \mathbb{R}^{n\times p}
$$
is the design matrix, ${\bm{\beta}} = (\beta_1,...,\beta_p)^{\prime} \in \mathbb{R}^p$ is the regression parameter vector, and ${\bm{\varepsilon}}_{n}=(\varepsilon_1,\varepsilon_2,...,\varepsilon_i,...,\varepsilon_n)^{\prime} \in \mathbb{R}^n$ is the error vector, respectively. Note that when $x_{i1}=1,i=1,...,n$, $\beta_1$ is then the intercept term. However, we do not make any distinction whether the linear model \eqref{lm} includes an intercept or not and continue to use this general form for convenience.

For illustrative purposes, we further assume that the design matrix $\mathbb{X}_n$ is full rank, that is, $\text{rank}(\mathbb{X}_n)=p, n\ge p+1$, and ${\bm{\varepsilon}}_{n} \sim N_n\left( {\bm{0}}_{n\times1}, \sigma^2\mathbb{I}_{n\times n} \right)$, that is, the error terms are independent and identically distributed (i.i.d.) $N(0,\sigma^2)$ random variables, where $\sigma^2$ is unknown to us. Under this widely-used Gauss-Markov setup with independent normal errors, a few well-known results are listed below for one's reference. 
\begin{enumerate}
	\item The least square estimator of $\bm{\beta}$ and its distribution:
	$$\hat{\bm{\beta}}_n = \left( \mathbb{X}_{n}^{\prime}\mathbb{X}_n \right)^{-1}\mathbb{X}_{n}^{\prime}{\bm{Y}}_n \sim N_p\left( \bm{\beta}, \sigma^2\left( \mathbb{X}_{n}^{\prime}\mathbb{X}_n \right)^{-1} \right).$$
	\item The estimator of $\sigma^2$ and its distribution:
	$$\hat{\sigma}^2 = S^2_n = (n-p)^{-1}\left( {\bm{Y}}_n - \mathbb{X}_n\hat{\bm{\beta}}_n \right)^{\prime}\left( {\bm{Y}}_n - \mathbb{X}_n\hat{\bm{\beta}}_n \right),$$
	and $$(n-p)S_n^2 \sim \sigma^2\chi^2_{n-p}.$$
	\item $\hat{\bm{\beta}}_n$ and $S_n^2$ are the uniformly minimum variance unbiased estimators (UMVUEs) of $\bm{\beta}$ and $\sigma^2$, respectively.
\end{enumerate}

Next, in the spirit of \cite{Mukhopadhyay (1974)} we introduce a loss function based on the margin of estimation errors:
\begin{equation}\label{loss}
L\left( \bm{\beta},\hat{\bm{\beta}}_n \right) = n^{-1}\left(\hat{\bm{\beta}}_n-\bm{\beta}\right)^{\prime}\left( \mathbb{X}_{n}^{\prime}\mathbb{X}_n \right)^{-1}\left(\hat{\bm{\beta}}_n-\bm{\beta}\right).
\end{equation}
The associated risk is accordingly given by
\begin{equation}\label{risk}
\begin{split}
\mathsf{E}\left[ L\left( \bm{\beta},\hat{\bm{\beta}}_n \right) \right] &= n^{-1}\mathsf{tr}\left\{\mathsf{E}\left[\left(\hat{\bm{\beta}}_n-\bm{\beta}\right)^{\prime}\left( \mathbb{X}_{n}^{\prime}\mathbb{X}_n \right)^{-1}\left(\hat{\bm{\beta}}_n-\bm{\beta}\right)\right]\right\}\\
&= n^{-1}\mathsf{E}\left\{\mathsf{tr}\left[\left( \mathbb{X}_{n}^{\prime}\mathbb{X}_n \right)^{-1}\left(\hat{\bm{\beta}}_n-\bm{\beta}\right)\left(\hat{\bm{\beta}}_n-\bm{\beta}\right)^{\prime}\right]\right\}\\
&= n^{-1}\mathsf{tr}\left\{\left( \mathbb{X}_{n}^{\prime}\mathbb{X}_n \right)^{-1}\mathsf{E}\left[\left(\hat{\bm{\beta}}_n-\bm{\beta}\right)\left(\hat{\bm{\beta}}_n-\bm{\beta}\right)^{\prime}\right]\right\}\\
&= n^{-1}\sigma^2\mathsf{tr}\left(\mathbb{I}_{p\times p}\right) = n^{-1}p\sigma^2.
\end{split}
\end{equation}
In this article, our goal is to make the risk not exceed a predetermined small level $b(>0)$, that is,
$$\mathsf{E}\left[ L\left( \bm{\beta},\hat{\bm{\beta}}_n \right) \right] = n^{-1}p\sigma^2 \le b.$$
We can then obtain the required minimum sample size 
\begin{equation} \label{opt}
n^* \equiv n^*(b) = b^{-1}p\sigma^2,
\end{equation} 
by tacitly disregarding the fact that $n^*$ may not be an integer. We define $n^*$ as the \textit{optimal sample size}, had $\sigma^2$ been known. However, since $\sigma^2$ is actually unknown to researchers and can be arbitrarily large, there exists no fixed-sample-size procedure that can bound the risk as desired. Consequently, sequential learning methods are necessary, where we estimate $\sigma^2$ by updating its UMVUE, $S^2_n$, at every stage as needed. 

In light of \cite{Hu and Zhuang (2022)}, we propose a sequential learning procedure $\mathcal{P}(\rho,k)$, which is efficient and enjoys operational convenience:
\begin{equation}\label{proc}
\begin{split}
T_{\mathcal{P}(\rho,k)} &\equiv T_{\mathcal{P}(\rho,k)}(b)=\inf \left\{n\geq0:m+kn\ge \rho b^{-1}pS^2_{m+kn}\right\},\\
N^*_{\mathcal{P}(\rho,k)} &\equiv N^*_{\mathcal{P}(\rho,k)}(b)=\rho^{-1}(m+kT_{\mathcal{P}(\rho ,k)}),\\
N_{\mathcal{P}(\rho,k)} &\equiv N_{\mathcal{P}(\rho,k)}(b)=\left\lfloor N^*_{\mathcal{P}(\rho,k)} \right\rfloor +1.
\end{split}
\end{equation}
Here, $0<\rho \le 1$ indicates a prefixed proportion, $k\ge 1$ indicates the number of observations taken at a time successively in the sequential sampling stage, $m\ge p+1$ indicates a pilot sample size picked in such a way that $m-p\equiv 0{\pmod k}$, and $\left\lfloor u\right\rfloor$ represents the largest integer that is strictly smaller than $u$. Denoting that $m-p=m_{0}k$ for some positive integer $m_{0}$, we further assume that the following limit operations hold:
\begin{equation}\label{limop}
m_{0}\rightarrow \infty, m=m_{0}k+p \rightarrow \infty, b\equiv b(m)=O(m^{-r}), n^{*}=O(m^{r}), \text{ and }\lim\sup\frac{m}{n^{*}}<\rho,
\end{equation}
where $r>1$ is a fixed constant. 

The sequential learning procedure $\mathcal{P}(\rho,k)$ given in \eqref{proc} is implemented as follows: Starting with a pilot sample of $m(=m_{0}k+p)$ observations, $(Y_{i},x_{i1},...,x_{ip}),i=1,...,m$, we sample $k$ observations at a time sequentially as needed and determine $T_{\mathcal{P}(\rho,k)}$. Next, we compute $N^*_{\mathcal{P}(\rho,k)}$ and $N_{\mathcal{P}(\rho,k)}$, and take additional observations as needed in one batch to make good the shortfall of this projected total sample size. It is clear that $\mathsf{P}\left(N_{\mathcal{P}(\rho,k)}<\infty \right)=1$, so this procedure terminates with probability one (w.p.1). Upon termination with the fully gathered data
\begin{equation*}
\left( Y_{i},x_{i1},...,x_{ip}\right), i =1,...,m,...,N_{\mathcal{P}(\rho,k)},
\end{equation*}
we compute $\hat{\bm{\beta}} \equiv \hat{\bm{\beta}}_{N_{\mathcal{P}(\rho,k)}}$, the least square estimate of the regression parameters, and construct the linear model
$\hat{y}=\bm{X}\hat{\bm{\beta}}$.

Observe that $N_{\mathcal{P}(\rho,k)} \uparrow \infty $ w.p.1 as $b\downarrow 0$. We are now in a position to state the efficiency properties of this newly developed sequential learning procedure in the following theorem.
\begin{theorem}\label{Thm}
For the sequential learning procedure $\mathcal{P}(\rho,k)$ given in \eqref{proc}, under the limit operations \eqref{limop} we have: 
\begin{equation}\label{eff}
\mathsf{E}\left[ N^*_{\mathcal{P}(\rho,k)} - n^* \right] = \rho^{-1}\eta(k)+o(1),
\end{equation}
where $\eta(k)=\frac{k-2}{2}-\sum_{i=1}^{\infty}n^{-1}\mathsf{E}\left[\left\{\chi^2_{kn}-2kn\right\}^{+}\right]$, and the average achieved risk
\begin{equation}\label{regret}
\mathsf{E}\left[ L\left( \bm{\beta},\hat{\bm{\beta}}_{N_{\mathcal{P}(\rho,k)}} \right)  \right] = b+o(b).
\end{equation}
\end{theorem} 

It is worth mentioning that for any fixed positive integer $k$, $\eta(k)$ in Theorem \ref{Thm} is a constant but cannot be evaluated in finite terms. We set out to approximate its value numerically by writing our own R codes.\footnote{The codes are available upon request.} In the spirits of \citet[Table 3.8.1]{Mukhopadhyay and Solanky (1994)} and \citet[Table 1]{Hu and Zhuang (2022)}, any term whose value is smaller than $10^{-15}$ in magnitude is dropped from the calculation of the infinite sum in $\eta(k)$. A few values of $\eta(k)$ are provided in Table \ref{Table 1}. 

\begin{table}[h!] 
\footnotesize
\captionsetup{font=footnotesize}
\caption{$\eta(k)$ approximations in Theorem \ref{Thm}}
\label{Table 1}\par
\centerline{\tabcolsep=3truept
\begin{tabular}{crccr}
\hline
\multicolumn{1}{c}{$k$} & \multicolumn{1}{c}{$\eta(k)$} &    ~~~   & \multicolumn{1}{c}{$k$} & \multicolumn{1}{c}{$\eta(k)$} \\ 
\hline
$1$ & \multicolumn{1}{r}{$-1.1826$} & & $11$ & $4.4204$ \\ 
$2$ & \multicolumn{1}{r}{$-0.5103$} & & $12$ & $4.9334$ \\ 
$3$ & \multicolumn{1}{r}{$0.1045$} & & $13$ & $5.4442$ \\ 
$4$ & \multicolumn{1}{r}{$0.6866$} & & $14$ & $5.9531$ \\ 
$5$ & \multicolumn{1}{r}{$1.2482$} & & $15$ & $6.4606$ \\ 
$6$ & \multicolumn{1}{r}{$1.7951$} & & $16$ & $6.9669$  \\ 
$7$ & \multicolumn{1}{r}{$2.3321$} & & $17$ & $7.4722$ \\ 
$8$ & \multicolumn{1}{r}{$2.8616$} & & $18$ & $7.9765$ \\ 
$9$ & \multicolumn{1}{r}{$3.3853$} & & $19$ & $8.4802$ \\ 
$10$ & \multicolumn{1}{r}{$3.9047$} & & $ 20$ & $8.9833$ \\ 
\hline
\end{tabular}}
\end{table}

Equation \eqref{eff} shows that for a sufficiently small $b$, the difference between the expected final sample size and the optimal sample size $n^*$ is about $\rho^{-1}\eta(k)$. This demonstrates the efficiency of our sequential learning procedure, for there exists no serious oversampling problem. Equation \eqref{regret} provides us with reasonable assurances that the achieved risk is expected to approach the target $b$ with a difference up to a small term given by $o(b)$. 



\setcounter{equation}{0}
\section{Simulations}\label{Sect. 3}

In this section, we conduct a series of Monte Carlo simulations to validate our theoretical results of the proposed sequential learning procedure as per \eqref{proc}.

First of all, we set up our linear regression model:  
$$Y=100-4X_1+3X_2+2X_3+\varepsilon,$$ 
where $X_i,i=1,2,3$, are explanatory variables whose values are generated from various normal distributions, that is, $N(\mu_i,\sigma^2_i)$ with $\mu_i$ being the mean value and $\sigma_i$ being the standard deviation. In addition, the error term $\varepsilon$ also follows a normal distribution but with the mean value being $0$ specifically, that is, $N(0,\sigma^2)$. We have conducted a comprehensive simulation study and tried different combinations for the values of $\mu_i$, $\sigma_i$, and $\sigma$. The results are very consistent across the board. One may also note that for brevity alone, we only illustrate this three-predictor linear regression model with the given coefficients values. The results of simulations for linear regression models with other coefficients values performed similarly. Thus, in this paper, we only report the simulated results from the following distributions: $X_1\sim N(50,9)$, $X_2\sim N(200,64)$, $X_3\sim N(100,25)$, and $\varepsilon \sim N(0,2)$. Moreover, we have considered a wide range of $b$ values including $\{0.4, 0.2, 0.1, 0.08, 0.04,0.02,0.01\}$ to show that the procedure works well for small, moderate, and large sample sizes. We demonstrate the simulated results with fixed $k=5,10,20$, $\rho=0.8$, and $m_0=2$. The whole sequential learning procedure for estimation is replicated $R=10,000$ times.

In \eqref{form}, we include a list of explanations for the main statistics in Table \ref{sim_results} as a reference.

\begin{equation}\label{form}
\begin{tabular}{cl}
\hline
$n^*$ & the required minimal sample size given that $\sigma$ is known \eqref{opt}\\
$N_{i}$ & the estimated final sample size from procedure \eqref{proc} for replication $i$, $i=1,...,R$\\
$\bar{N}$ & $=R^{-1}\Sigma_{i=1}^{R}N_{i}$, an estimate of $n^*$\\
$s.e_{\bar{N}}$ & $=\left\{  \tfrac{1}{(R^{2}-R)}\Sigma_{i=1}^{R}(N_{i}-\bar{N})^{2}\right\}^{1/2}$, the estimated standard error of $\bar{N}$\\
$\bar{\hat{\sigma}}$ & $=R^{-1}\Sigma_{i=1}^{R}\hat{\sigma}_{i}$, an estimate of $\sigma$\\
$s.e._{\hat{\sigma}}$ & $=\left\{  \tfrac{1}{(R^{2}-R)}\Sigma_{i=1}^{R}(\hat{\sigma}_{i}-\bar{\hat{\sigma}})^{2}\right\}^{1/2}$, the estimated standard error of $\bar{\hat{\sigma}}$\\
$r^*$ & $=\mathsf{E}\left[ L\left( \bm{\beta},\hat{\bm{\beta}}_n \right) \right]$, as per \eqref{risk}\\
$r_i$ & the achieved risk as per \eqref{loss} for repliaction $i$, $i=1,...,R$\\ 
$\bar{r}$ & $=R^{-1}\sum_{i=1}^{R}r_i$, an estimate of $r^*$ \\
$s.e._{\bar{r}}$ & $=\left\{  \frac{1}{(R^{2}-R)}\Sigma_{i=1}^{R}({r}_{i}-\bar{{r}})^{2}\right\}^{1/2}$, the estimated standard error of $\bar{r}$\\
\hline
\end{tabular}
\end{equation}

From Table \ref{sim_results}, one can see that as $b$ gets smaller, $\bar{N}/n^*$ gets closer to $1$. Moreover, $\bar{N}-n^*$ hangs around its corresponding $\rho^{-1}\eta(k)$ value across the board. The mean value of the estimated $\sigma$ hangs tightly around its true value $2$ with small standard errors. The estimated risks from different scenarios are always close to the optimal risk as per \eqref{risk} with very small standard errors.

\begin{table}[p]
  \centering
  \caption{Simulated results implementing the sequential learning procedure \eqref{proc} with $\rho=0.8$, $m_0=2$ under 10,000 runs}
   \smallskip
   \resizebox{\columnwidth}{!}{%
    \begin{tabular}{ccccccccccc}
    \hline
          &       &       & \multicolumn{6}{c}{$k=5,\rho^{-1}\eta(5)=1.560$}                       &       &  \\
    \multicolumn{1}{c}{$b$} & \multicolumn{1}{c}{$n^*$} & \multicolumn{1}{c}{$\bar{N}$} & \multicolumn{1}{c}{${s.e.}_{\bar{N}}$} &    \multicolumn{1}{c}{$\bar{N}/n^*$}    &  \multicolumn{1}{c}{$\bar{N}-n^*$}       & \multicolumn{1}{c}{$\bar{\hat{\sigma}}$} & \multicolumn{1}{c}{${s.e.}_{\hat{\sigma}}$} & \multicolumn{1}{c}{$r^*$} & \multicolumn{1}{c}{$\bar{r}$} & \multicolumn{1}{c}{${s.e.}_{\bar{r}}$} \\
    \hline
    0.4   & 40    & 41.089 & 0.117 & 1.027 & 1.089 & 1.9194 & 0.0027 & 0.40  & 0.374 & 0.0008 \\
    0.2   & 80    & 81.044 & 0.163 & 1.013 & 1.044 & 1.9600 & 0.0019 & 0.20  & 0.194 & 0.0003 \\
    0.1   & 160   & 161.452 & 0.209 & 1.009 & 1.452 & 1.9823 & 0.0012 & 0.10  & 0.098 & 0.0001 \\
    0.08  & 200   & 202.256 & 0.235 & 1.011 & 2.256 & 1.9886 & 0.0010 & 0.08  & 0.079 & 0.0000 \\
    0.04  & 400   & 402.279 & 0.321 & 1.006 & 2.279 & 1.9948 & 0.0007 & 0.04  & 0.040 & 0.0000 \\
    0.02  & 800   & 802.372 & 0.453 & 1.003 & 2.372 & 1.9974 & 0.0005 & 0.02  & 0.020 & 0.0000 \\
    0.01  & 1600  & 1601.390 & 0.635 & 1.001 & 1.390 & 1.9981 & 0.0004 & 0.01  & 0.010 & 0.0000 \\\hline
         &       &       & \multicolumn{6}{c}{$k=10,\rho^{-1}\eta(10)=4.881$}                       &       &  \\
    \multicolumn{1}{c}{$b$} & \multicolumn{1}{c}{$n^*$} & \multicolumn{1}{c}{$\bar{N}$} & \multicolumn{1}{c}{$s.e._{\bar{N}}$} &    \multicolumn{1}{c}{$\bar{N}/n^*$}    &  \multicolumn{1}{c}{$\bar{N}-n^*$}       & \multicolumn{1}{c}{$\bar{\hat{\sigma}}$} & \multicolumn{1}{c}{$s.e._{\hat{\sigma}}$} & \multicolumn{1}{c}{$r^*$} & \multicolumn{1}{c}{$\bar{r}$} & \multicolumn{1}{c}{$s.e._{\bar{r}}$} \\
    \hline
    0.4   & 40    & 45.419 & 0.106 & 1.135 & 5.419 & 1.9460 & 0.0023 & 0.40  & 0.343 & 0.0006 \\
    0.2   & 80    & 84.519 & 0.160 & 1.056 & 4.519 & 1.9642 & 0.0017 & 0.20  & 0.186 & 0.0002 \\
    0.1   & 160   & 164.922 & 0.209 & 1.031 & 4.922 & 1.9845 & 0.0011 & 0.10  & 0.096 & 0.0001 \\
    0.08  & 200   & 204.933 & 0.238 & 1.025 & 4.933 & 1.9868 & 0.0010 & 0.08  & 0.077 & 0.0000 \\
    0.04  & 400   & 405.389 & 0.321 & 1.013 & 5.389 & 1.9945 & 0.0007 & 0.04  & 0.039 & 0.0000 \\
    0.02  & 800   & 804.934 & 0.453 & 1.006 & 4.934 & 1.9964 & 0.0005 & 0.02  & 0.020 & 0.0000 \\
    0.01  & 1600  & 1605.252 & 0.633 & 1.003 & 5.252 & 1.9985 & 0.0004 & 0.01  & 0.010 & 0.0000 \\\hline
        &       &       & \multicolumn{6}{c}{$k=20,\rho^{-1}\eta(20)=11.229$}                       &       &  \\
    \multicolumn{1}{c}{$b$} & \multicolumn{1}{c}{$n^*$} & \multicolumn{1}{c}{$\bar{N}$} & \multicolumn{1}{c}{$s.e._{\bar{N}}$} &    \multicolumn{1}{c}{$\bar{N}/n^*$}    &  \multicolumn{1}{c}{$\bar{N}-n^*$}       & \multicolumn{1}{c}{$\bar{\hat{\sigma}}$} & \multicolumn{1}{c}{$s.e._{\hat{\sigma}}$} & \multicolumn{1}{c}{$r^*$} & \multicolumn{1}{c}{$\bar{r}$} & \multicolumn{1}{c}{$s.e._{\bar{r}}$} \\
    \hline
    0.4   & 40    & 56.428 & 0.058 & 1.411 & 16.428 & 1.984 & 0.0020 & 0.40  & 0.282 & 0.0005 \\
    0.2   & 80    & 91.143 & 0.167 & 1.139 & 11.143 & 1.973 & 0.0020 & 0.20  & 0.174 & 0.0002 \\
    0.1   & 160   & 171.74 & 0.219 & 1.073 & 11.740 & 1.988 & 0.0010 & 0.10  & 0.093 & 0.0001 \\
    0.08  & 200   & 211.645 & 0.239 & 1.058 & 11.645 & 1.989 & 0.0010 & 0.08  & 0.075 & 0.0001 \\
    0.04  & 400   & 410.973 & 0.328 & 1.027 & 10.973 & 1.994 & 0.0010 & 0.04  & 0.039 & 0.0000 \\
    0.02  & 800   & 811.845 & 0.454 & 1.015 & 11.845 & 1.997 & 0.0010 & 0.02  & 0.020 & 0.0000 \\
    0.01  & 1600  & 1610.965 & 0.636 & 1.007 & 10.965 & 1.998 & 0.0000 & 0.01  & 0.010 & 0.0000 \\\hline
    \end{tabular}}
  \label{sim_results}%
\end{table}%


\setcounter{equation}{0}
\section{Online Sales Examination}\label{Sect. 4}

To illustrate the practical applicability of our newly proposed sequential learning procedure, we implement it to examine the influential factors of online sales using data from Tmall.com, one of the largest online platforms for local Chinese and international businesses to sell brand\zhao{-}name products to consumers in China. Tmall operates in a similar way as Amazon, except that it is a dedicated business-to-consumer (B2C) e-commerce website while Amazon is both a B2C and C2C (consumer-to-consumer) platform. We use part of the sales data collected by \cite{Huang and Pape (2020)} through a web scraping program. Collecting data from websites can be time-consuming and costly. A typical question faced by researchers is to determine an appropriate sample size for their analysis. This is a perfect case to illustrate our sequential learning procedure in determining when the web scraping can stop. For illustrative purposes, we focus on two random sellers of ten electronic appliances, and our sequential learning procedure will collect additional sales information first from more products (up to 10) and then along the time dimension. The data set consists of the following variables: $Y$, daily sales of a product from a seller; $X_1$, the product retail price; $X_2$, the lagged total reviews; $X_3$, the number of daily reviews; $X_4$, the average product rating; $X_5$, the seller-specific description rating; $X_6$, the seller-specific service rating; $X_7$, the seller-specific shipping rating; $X_8$, the user grade based on a consumer's past purchasing experiences and sellers' feedback ratings; $X_9$, the number of images posted in the daily reviews; $X_{10}$, the number of characters in the daily reviews; $X_{11}$, the number of follow-up reviews in the daily reviews; $X_{12}$, the number of characters in the follow-up reviews; and $D$, a dummy variable with $1$ indicating the first seller and $0$ otherwise. One can refer to \cite{Huang and Pape (2020)} for more background details.

Since the response variable $Y$ and the explanatory variables $X_i,i=1,2,...,12$ all take nonnegative values, we apply the Box-Cox shifted power transformation to reduce skewness: 
$$Y^\prime = \ln(Y+1), \text{ and } X_i^{\prime} = \ln(X_i+1), i = 1,2,...,12.$$
Then, we construct the following linear model for examining the daily sales:
\begin{equation}\label{lm}
Y^\prime = \beta_0 + \sum_{i=1}^{12}\beta_iX^\prime_i + \beta_{13}D + \varepsilon.
\end{equation}

In this scenario, $p=14$. Let $b=0.01$ be the target bounded risk. Since there are two sellers whose products are both recorded every day, we pick $k=2$ as the number of observations to be collected at a time in the sequential sampling stage. By further fixing $m_0=10$, we start with a pilot sample of size $m=m_0k+p = 34$. A small value of $\rho=0.5$ is chosen to facilitate the operation of the sequential learning procedure, which terminates with a final sample size of 156. The regression result is displayed in Table \ref{sales} column (1). The \textit{p}-values of the Shapiro-Wilk test on the standardized residuals is 0.08471, indicating that the normality assumption on error terms is not violated. The \textit{p}-value of the Breusch-Pagan test to assess homoscedasticity is 0.05885, which means we can still keep the homoscedasticity assumption with a close call. The Breusch-Godfrey test for autocorrelation returns a \textit{p}-value of 0.6025, so the error terms can be considered independent. These statistical tests validate our model assumptions.

According to the estimates in column (1), we see that a lower price level, more daily reviews, higher product ratings, and higher seller ratings tend to increase product sales. Note that the Tmall market is a highly competitive B2C market, which means that there are hundreds of sellers for one specific product, and the electronic products sold in this market have relatively high quality and are standardized. The market environment is therefore close to perfect competition. As a result, the explanatory variable price could be considered exogenous in the above regression, and the negative coefficient in front of it suggests that lower-priced products tend to have larger markets. Also, it is easy to observe that the linear model (1) includes a few explanatory variables that are statistically insignificant, such as the lagged total reviews, user grade, image posts, and review characters. For comparative purposes, we also construct some alternative linear models with fewer explanatory variables shown in columns (2) and (3) of Table \ref{sales}. We still have the same findings.

\begin{table}[h]
  \centering
  \caption{Regression results of linear models for examining online sales}
   \smallskip
   \resizebox{\columnwidth}{!}{%
    \begin{tabular}{llll}
    \hline
    \multicolumn{4}{l}{Response variable: $Y^\prime=\ln(\text{Daily sales}+1)$}\\
    \hline
    Explanatory variables & (1) & (2) & (3) \\
    \hline
    Intercept & $-83.03\threeS(20.58)$ & $-83.06\threeS(13.75)$ & $-81.49\threeS(13.61)$ \\
    $X_1^\prime=\ln(\text{Price}+1)$ & $-0.1933\oneS(0.1158)$ & $-0.1933\threeS(0.0721)$ & $-0.1894\threeS(0.0718)$\\
    $X_2^\prime=\ln(\text{Lagged total reviews}+1)$ & $-0.0087(0.0991)$ & & \\
    $X_3^\prime=\ln(\text{Daily reviews}+1)$ & $0.6284\threeS(0.0690)$ & $0.6267\threeS(0.0619)$ & $0.6600\threeS(0.0470)$\\
    $X_4^\prime=\ln(\text{Product rating}+1)$ & $35.37\threeS(8.807)$ & $35.34\threeS(6.543)$ & $34.15\threeS(6.375)$\\
    $X_5^\prime=\ln(\text{Description rating}+1)$ & $76.24\twoS(35.89)$ & $76.34\threeS(28.36)$ & $72.91\twoS(28.02)$\\
    $X_6^\prime=\ln(\text{Service rating}+1)$ & $-159.9\threeS(46.00)$ & $-158.3\threeS(38.20)$ & $-153.0\threeS(37.63)$\\
    $X_7^\prime=\ln(\text{Shipping rating}+1)$ & $96.26\threeS(18.73)$ & $94.61\threeS(16.63)$ & $93.09\threeS(16.51)$\\
    $X_8^\prime=\ln(\text{User grade}+1)$ & $0.0210(0.0792)$ & $0.0320(0.0387)$ & \\
    $X_9^\prime=\ln(\text{Image posts}+1)$ & $-0.0128(0.0519)$ & & \\
    $X_{10}^\prime=\ln(\text{Follow-up reviews}+1)$ & $0.0313(0.221)$ & & \\
    $X_{11}^\prime=\ln(\text{Review characters}+1)$ & $0.0085(0.0587)$ & & \\
    $X_{12}^\prime=\ln(\text{Follow-up rev. char.}+1)$ & $0.0056(0.0609)$ & & \\
    $D=1(0)$: Seller 1(0) & $0.5415\threeS(0.1737)$ & $0.5304\threeS(0.1098)$ & $0.5277\threeS(0.1096)$\\
    $R^2$ & 0.7819 & 0.7813 & 0.7803\\
    Adjusted $R^2$ & 0.7620 & 0.7694 & 0.7699 \\
    \hline
    \multicolumn{4}{l}{$\threeS p<0.01; \twoS p<0.05; \oneS p<0.1$. Standard errors are reported in parentheses.}
    \end{tabular}}
  \label{sales}%
\end{table}%

One may feel surprised that while the service rating variable $X_6^\prime$ is statistically significant, the sign turns out to be negative, indicating that better service leads to fewer sales. This counterintuitive estimate is due to the issue of multicollinearity. The following correlation matrix in Table \ref{cormat} shows that the three seller-specific variables, i.e., description rating variable $X_5^\prime$, service rating variable $X_6^\prime$, and shipping rating variable $X_7^\prime$, are highly correlated. 

\begin{table}[h]
  \centering
  \caption{Correlation matrix of the seller-specific rating variables}
   \smallskip
   {\footnotesize
   \begin{tabular}{|c|c|c|c|}
   \hline
   	& $X_5^\prime$ & $X_6^\prime$ & $X_7^\prime$\\
   \hline
   $X_5^\prime$ & 1 & 0.9727 & 0.9647\\
   \hline
   $X_6^\prime$ & 0.9727 & 1 & 0.9307\\
   \hline
   $X_7^\prime$ & 0.9647 & 0.9307 & 1\\
   \hline
   \end{tabular}
   }
  \label{cormat}%
\end{table}%

Along the lines of \cite{Huang and Pape (2020)}, to avoid the multicollinearity, we include only one of the three seller-specific rating variables at a time into the linear model of column (3) in Table \ref{sales}. The corresponding regression results are displayed in Table \ref{sales2}. Now, we observe that three seller-specific rating variables in models (4)-(6) all have positive signs. However, only the shipping rating remains significant. This may indicate that consumers are more concerned about the sellers' shipping speed rather than the sellers' description or service. As the $R^2$ value does not decrease dramatically compared with the full-specification model in column (1) of Table \ref{sales}, we suggest using linear model (6) to predict online sales for simplicity.  

\begin{table}[h]
  \centering
  \caption{Regression results of reduced linear models for examining online sales}
   \smallskip
   \resizebox{\columnwidth}{!}{%
    \begin{tabular}{llll}
    \hline
    \multicolumn{4}{l}{Response variable: $Y^\prime=\ln(\text{Daily sales}+1)$}\\
    \hline
    Explanatory variables & (4) & (5) & (6) \\
    \hline
    Intercept & $-47.02\threeS(9.047)$ & $-45.84\threeS(8.315)$ & $-49.93\threeS(8.263)$ \\
    $X_1^\prime=\ln(\text{Price}+1)$ & $-0.2727\threeS(0.0772)$ & $-0.2701\threeS(0.0770)$ & $-0.2434\threeS(0.0762)$\\
    $X_3^\prime=\ln(\text{Daily reviews}+1)$ & $0.7286\threeS(0.0437)$ & $0.7350\threeS(0.0445)$ & $0.7473\threeS(0.0434)$\\
    $X_4^\prime=\ln(\text{Product rating}+1)$ & $23.63\threeS(5.661)$ & $22.29\threeS(6.006)$ & $18.10\threeS(5.644)$\\
    $X_5^\prime=\ln(\text{Description rating}+1)$ & $4.706(5.393)$ &  &  \\
    $X_6^\prime=\ln(\text{Service rating}+1)$ &  & $5.382(4.927)$ & \\
    $X_7^\prime=\ln(\text{Shipping rating}+1)$ &  &  & $11.75\twoS(4.574)$\\
    $D=1(0)$: Seller 1(0) & $0.2190\threeS(0.0814)$ & $0.5304\threeS(0.1098)$ & $0.2629\threeS(0.0817)$\\
    $R^2$ & 0.7310 & 0.7318 & 0.7410\\
    Adjusted $R^2$ & 0.7220 & 0.7228 & 0.7324 \\
    \hline
    \multicolumn{4}{l}{$\threeS p<0.01; \twoS p<0.05; \oneS p<0.1$. Standard errors are reported in parentheses.}
    \end{tabular}}
  \label{sales2}%
\end{table}%


\setcounter{equation}{0}
\section{Proofs}\label{Sect. 5}

In this section, we lay out the proof of Theorem \ref{Thm}. Note that the stopping time $T_{\mathcal{P}(\rho,k)}$ from \eqref{proc} can be rewritten as
\begin{equation*}
\begin{split}
T_{\mathcal{P}(\rho,k)} &= \inf\left\{n\ge0: (m+kn)(m+kn-p)\ge \rho b^{-1}p \sigma^2 (m+kn-p)\sigma^{-2}S^2_{m+kn}\right\}\\
& = \inf\left\{ n\ge0: k(n+m_0)\left[k(n+m_0)-p\right] \ge \rho n^* \sum_{i=1}^{n+m_0}U_i \right\}.
\end{split}
\end{equation*}
Defining $t=T_{\mathcal{P}(\rho,k)}+m_0$, we have
\begin{equation}\label{t}
\begin{split} 
t & = \inf\left\{ n\ge m_0: kn\left(kn-p\right) \ge \rho n^* \sum_{i=1}^{n}U_i \right\}\\
& = \inf\left\{ n\ge m_0: (kn)^{-1}\sum_{i=1}^{n}U_i \le (kn/\rho n^*)(1-p(kn)^{-1}) \right\},
\end{split}
\end{equation}
where $U_i\equiv\sum_{j=(i-1)k+1}^{ik}W_j,i=1,...,n$, and $W_j,j=1,...,kn$ are i.i.d. $\chi^2_1$ random variables. Now, the stopping time $t$ in \eqref{t} has the same form as (2.1) of \cite{Hu and Zhuang (2022)}. By their Theorem 2.1, under the limit operations \eqref{limop} we have that
\begin{equation}\label{Et}
\begin{split}
\mathsf{E}[kt-\rho n^*] &= \frac{k}{2}-1-p-\sum_{n=1}^{\infty}n^{-1}\mathsf{E}\left[ \left\{ \chi^2_{kn}-2kn \right\}^{+}\right]+o(1)\\
& = \eta(k)-p+o(1).
\end{split}
\end{equation}
Since $N^*_{\mathcal{P}(\rho,k)} = \rho^{-1}[m+k(t-m_0)] = \rho^{-1}(kt+p)$, then \eqref{Et} leads to
\begin{equation*}
\begin{split}
\mathsf{E}\left[ N^*_{\mathcal{P}(\rho,k)}-n^* \right] = \rho^{-1}\eta(k)+o(1).
\end{split}
\end{equation*}
One may refer to \cite{Woodroofe (1977)} or Section A.4 of the Appendix in \cite{Mukhopadhyay and de Silva (2009)} for more details, as well. The proof of \eqref{eff} is complete.

Next, we set out to show \eqref{regret}. In light of \cite{Mukhopadhyay (1974)}, we have the following crucial lemma. Since a proof of the lemma can be found in \citet[Section 12.6.1]{Mukhopadhyay and de Silva (2009)}, we leave out many details for brevity. 
\begin{lemma}\label{lemma1}
For all fixed $n \ge m$, $\bm{\hat{\beta}}_n$ and $(S_m^2,...,S_n^2)$ are independent.
\end{lemma}

Observe that in terms of the sequential learning procedure \eqref{proc}, for all $n\ge m$, the event $\mathsf{I}\{N_{\mathcal{P}(\rho,k)}=n\}$ depends on $(S_m^2,...,S_n^2)$ alone, where $\mathsf{I}\{ A \}$ is the indicator function of an event $A$. By Lemma \ref{lemma1}, therefore, $\mathsf{I}\{N_{\mathcal{P}(\rho,k)}=n\}$ and $\bm{\hat{\beta}}_n$ are independent. Then, the achieved risk can be written as
\begin{equation}\label{AchRisk}
\begin{split}
& \mathsf{E}\left[ L\left( \bm{\beta},\hat{\bm{\beta}}_{N_{\mathcal{P}(\rho,k)}} \right)  \right]\\
= \ & \sum_{n=m}^{\infty} \mathsf{E}\left[ \left. L\left( \bm{\beta},\hat{\bm{\beta}}_{N_{\mathcal{P}(\rho,k)}} \right) \right| N_{\mathcal{P}(\rho,k)}=n \right]\mathsf{P}\left(N_{\mathcal{P}(\rho,k)}=n\right)\\
= \ & \sum_{n=m}^{\infty} \mathsf{E}\left[ L\left( \bm{\beta},\hat{\bm{\beta}}_{n} \right) \right]\mathsf{P}\left(N_{\mathcal{P}(\rho,k)}=n\right)\\
= \ & \sum_{n=m}^{\infty} n^{-1}p\sigma^2\mathsf{P}\left(N_{\mathcal{P}(\rho,k)}=n\right) = p\sigma^2\mathsf{E}\left[N_{\mathcal{P}(\rho,k)}^{-1}\right] = b\mathsf{E}\left[n^{*}N_{\mathcal{P}(\rho,k)}^{-1}\right]\\
= \ & b\mathsf{E}\left[n^{*}N_{\mathcal{P}(\rho,k)}^{-1}\mathsf{I}\left\{N_{\mathcal{P}(\rho,k)} \le \frac{1}{2}n^* \right\}\right] + b\mathsf{E}\left[n^{*}N_{\mathcal{P}(\rho,k)}^{-1}\mathsf{I}\left\{N_{\mathcal{P}(\rho,k)} > \frac{1}{2}n^* \right\}\right]. 
\end{split}
\end{equation} 

To evaluate $\mathsf{E}\left[n^{*}N_{\mathcal{P}(\rho,k)}^{-1}\mathsf{I}\left\{N_{\mathcal{P}(\rho,k)} \le \frac{1}{2}n^* \right\}\right]$, we introduce a useful lemma below.
\begin{lemma}\label{lemma2}
For the sequential learning procedure $\mathcal{P}(\rho,k)$ given in \eqref{proc}, under the limit operations \eqref{limop} we have: 
\begin{equation*}
\mathsf{P}\left( N_{\mathcal{P}(\rho,k)} \le \gamma n^*  \right) = O(n^{*-\frac{s}{2r}}),
\end{equation*}
for any $s \ge 2$, where $0<\gamma<1$ is a fixed proportion.
\end{lemma}
\begin{proof}
Note that on the set $\left\{ N_{\mathcal{P}(\rho,k)} \le \gamma n^* \right\}$, we have
\begin{equation*}
\begin{split}
N^*_{\mathcal{P}(\rho,k)} \le \gamma n^*\Rightarrow m+kT_{\mathcal{P}(\rho,k)} \le \left\lfloor \gamma \rho n^* \right\rfloor + 1 = t_u, \ \text{say.} 
\end{split}
\end{equation*}
Applying Kolmogorov's inequality for reverse martingales yields 
\begin{equation*}
\begin{split}
& \mathsf{P}\left( N_{\mathcal{P}(\rho,k)} \le \gamma n^*  \right) \le \mathsf{P}\left( m+kT_{\mathcal{P}(\rho,k)} \le \gamma n^* \right) \\
\le \ & \mathsf{P}\left( S^2_{m+kT_{\mathcal{P}(\rho,k)}} \le \gamma \sigma^2  \right) \le \mathsf{P}\left( \left| S^2_{m+kT_{\mathcal{P}(\rho,k)}}-\sigma^2 \right| \ge (1-\gamma) \sigma^2  \right)\\
\le \ & \mathsf{P}\left( \max_{m \le t \le t_u} \left| S^2_{t}-\sigma^2 \right| \ge (1-\gamma)\sigma^2  \right)\\
\le \ & \left[(1-\gamma)\sigma^2\right]^{-s}\mathsf{E}\left|S_m^2-\sigma^2\right|^s = O(m^{-\frac{s}{2}}) = O(n^{*-\frac{s}{2r}}),
\end{split}
\end{equation*}
for any $s \ge 2$.
\end{proof}

By Lemma \ref{lemma2}, we obtain
\begin{equation}\label{I}
\begin{split}
b\mathsf{E}\left[n^{*}N_{\mathcal{P}(\rho,k)}^{-1}\mathsf{I}\left\{N_{\mathcal{P}(\rho,k)} \le \frac{1}{2}n^* \right\}\right] \le \frac{bn^*}{m}\mathsf{P}\left(N_{\mathcal{P}(\rho,k)} \le \frac{1}{2}n^*\right)=O(n^{*-\frac{s+2}{2r}})=o(b),
\end{split}
\end{equation}
by selecting appropriate $s>2(r-1)$. To evaluate $\mathsf{E}\left[n^{*}N_{\mathcal{P}(\rho,k)}^{-1}\mathsf{I}\left\{N_{\mathcal{P}(\rho,k)} > \frac{1}{2}n^* \right\}\right]$, we focus on the asymptotic behavior of $n^{*}N_{\mathcal{P}(\rho,k)}^{-1}$. Recall the stopping rule defined in \eqref{proc}. On the one hand, we have
\begin{equation}\label{II1}
\begin{split}
N_{\mathcal{P}(\rho,k)} \ge N^*_{\mathcal{P}(\rho,k)} \ge b^{-1}pS^2_{m+kT_{\mathcal{P}(\rho,k)}} \Rightarrow n^{*}N^{-1}_{\mathcal{P}(\rho,k)} \le \sigma^2S^{-2}_{m+kT_{\mathcal{P}(\rho,k)}},
\end{split}
\end{equation}
which converges to 1 in probability as $b \to 0$. On the other hand,
\begin{equation}\label{II2}
\begin{split}
& N_{\mathcal{P}(\rho,k)} \le N^*_{\mathcal{P}(\rho,k)}+1 \le \rho^{-1}\left( k+\rho b^{-1}pS^2_{m+k\left(T_{\mathcal{P}(\rho,k)}-1\right)} \right) + 1\\
\Rightarrow \ & n^*N^{-1}_{\mathcal{P}(\rho,k)} \ge \sigma^2\left( S^2_{m+k\left(T_{\mathcal{P}(\rho,k)}-1\right)} + bp^{-1}(\rho^{-1}k+1) \right)^{-1}, 
\end{split}
\end{equation} 
which also converges to 1 in probability as $b \to 0$. Combining \eqref{II1}, \eqref{II2}, and Lemma \ref{lemma2}, we can claim that $n^{*}N_{\mathcal{P}(\rho,k)}^{-1}\mathsf{I}\left\{N_{\mathcal{P}(\rho,k)} > \frac{1}{2}n^* \right\}$ converges to 1 in probability as $b \to 0$. Since it is true that
$n^{*}N_{\mathcal{P}(\rho,k)}^{-1}\mathsf{I}\left\{N_{\mathcal{P}(\rho,k)} > \frac{1}{2}n^* \right\} \le 2$, the dominated convergence theorem guarantees that as $b \to 0$,
\begin{equation} \label{II}
\mathsf{E}\left[n^{*}N_{\mathcal{P}(\rho,k)}^{-1}\mathsf{I}\left\{N_{\mathcal{P}(\rho,k)} > \frac{1}{2}n^* \right\}\right] = 1 + o(1)
\end{equation}
Putting together \eqref{AchRisk}, \eqref{I} and \eqref{II}, we obtain the desired result \eqref{regret}.


\setcounter{equation}{0}
\section{Concluding Remarks}\label{Sect. 6}

Linear regression models are one of the most widely used statistical models with a vast of applications. Accurately estimating the parameters for the models is important since it affects how one would explain the effect from an explanatory variable to the response variable. In real life, it is not often the case that one would have the data ahead of time. On another note, researchers or practitioners would not wait an infinite amount of time collecting data or decide randomly whether the data are enough. Instead, one may seek a valid statistical learning approach to getting to a point where the sample is large enough to conduct accurate estimations. The proposed sequential learning procedure addresses this question. We assume that one would like to limit the risk of estimation and prefer the minimum required sample size. Our proposed procedure provides a strategic sampling plan that incorporates a stopping rule while it considers the estimation accuracy, sampling cost, and sampling logistics one would desire. It is proven to enjoy efficiency properties. And it is demonstrated using simulation studies and a real data example.

In this paper, we only consider the classic linear models under the Gauss-Markov setup with independent normal errors. Further research may extend to parameters estimation of a logistic regression model or even a more generalized linear regression model such as Poisson regression or negative binomial regression.


\end{document}